%% file: main.tex
\documentclass[journal]{new-aiaa}
\usepackage[utf8]{inputenc}
\usepackage{dsfont}
\usepackage[header,page,title]{appendix}
\usepackage{multirow}
\usepackage{graphicx}
\usepackage{float}

\usepackage{amsmath,bm,bbm,latexsym,subcaption,epstopdf}
\usepackage[version=4]{mhchem}
\usepackage[title]{appendix}
\usepackage{longtable,tabularx}
\usepackage[ruled,linesnumbered]{algorithm2e}
\graphicspath{{./figures/}} 

\DeclareGraphicsExtensions{.pdf,.eps}

\input{commands}

\title{Bang--Bang Evasion: Its Stochastic Optimality \\ and a Terminal-Set-Based Implementation}

\author{Liraz Mudrik\footnote{Visiting Lecturer, Stephen B.\ Klein Faculty of Aerospace
      Engineering. This work is part of Dr.\ Mudrik's Doctoral Research
      at the Technion; \texttt{liraz109@gmail.com}. Member AIAA
      (Corresponding Author).}}
\author{Yaakov Oshman\footnote{Professor Emeritus, Stephen B.\ Klein
    Faculty of Aerospace Engineering; \texttt{yaakov.oshman@technion.ac.il}.
    Fellow AIAA.}}
\affil{Technion---Israel Institute of Technology, 3200003 Haifa, Israel}

\begin{document}

\maketitle

\ifmarkup\linenumbers\fi

\renewcommand{\include}{\input}

\include{abstract}
\include{nomenclature}

\include{intro}

\include{probdef}
\include{result}
\include{terminal_set_LM}

\include{simulation}
\include{conclusions}

\section*{Acknowledgment}

The authors would like to thank Prof.\ Shimrit Shtern of the
Technion---Israel Institute of Technology for her helpful insights
during the early stages of this research.

\bibliography{Bib1}

\end{document}

%% file: commands.tex
\newif\ifmarkup
\markupfalse
\usepackage{xcolor}
\makeatletter
\ifmarkup
  \usepackage[pagewise]{lineno}
  \newcommand*\@patchAmsLineno[1]{%
    \expandafter\let\csname old#1\expandafter\endcsname\csname #1\endcsname
    \expandafter\let\csname oldend#1\expandafter\endcsname\csname end#1\endcsname
    \renewenvironment{#1}%
      {\linenomath\csname old#1\endcsname}%
      {\csname oldend#1\endcsname\endlinenomath}}%
  \newcommand*\@patchAmsLinenoBoth[1]{\@patchAmsLineno{#1}\@patchAmsLineno{#1*}}%
  \AtBeginDocument{%
    \@patchAmsLinenoBoth{equation}%
    \@patchAmsLinenoBoth{align}%
    \@patchAmsLinenoBoth{flalign}%
    \@patchAmsLinenoBoth{alignat}%
    \@patchAmsLinenoBoth{gather}%
    \@patchAmsLinenoBoth{multline}}%
  \newcommand{\rev}[1]{\textcolor{blue}{#1}}
  \newcommand{\revblue}{\color{blue}}
\else
  \newcommand{\rev}[1]{#1}%
  \newcommand{\revblue}{}%
\fi
\makeatother

\newlength{\figwidth}
\newcommand{\dfn}{\triangleq}
\newcommand{\given}{\mid}

\DeclareMathOperator{\sign}{sgn}

\DeclareMathOperator{\expect}{\mathbb{E}}

\newcommand{\meas}{\ensuremath{\mathcal{Y}}}

\let\abs=\envert

\usepackage{amsthm}

\newtheorem{exmp}{Example}[section]
\newtheorem{theorem}{Theorem}[section]

\newtheorem{lemma}[theorem]{Lemma}

%% file: abstract.tex
\begin{abstract}
We address the problem of optimal evasion in a planar endgame engagement, where a target with bounded lateral acceleration seeks to avoid interception by a missile guided by a linear feedback law.
Contrary to existing approaches, that assume perfect information or use heuristic maneuver models in stochastic settings, we formulate the problem in an inherently stochastic framework involving imperfect information and bounded controls.  
Complying with the generalized separation theorem, the control law factors in the posterior distribution of the state.  
\rev{We extend} the well-known optimality of bang-bang evasion maneuvers in deterministic settings to the realm of realistic, stochastic evasion scenarios\rev{. First, we prove} that an optimal evasion strategy always exists, and that the set of optimal solutions includes at least one bang-bang policy, rendering the resulting optimal control problem finite-dimensional.
\rev{Second, leveraging this structure, we propose} the closed-loop terminal-set-based evasion (TSE) strategy, and demonstrate its effectiveness in simulation against a proportional navigation
pursuer.
Monte Carlo simulations show that the TSE strategy outperforms traditional stochastic evasion strategies based on random telegraph, Singer, and weaving models.
\end{abstract}


%% file: nomenclature.tex
{\revblue
\section*{Nomenclature}
\noindent\begin{tabular}{@{}l@{\ \ =\ \ }p{0.72\columnwidth}@{}}
$\grave{a}_{M},\ \grave{a}_{T}$ & interceptor and target accelerations normal to the initial LOS \\
$\mathbf{a}^{n}(i,j),\ \mathbf{z}^{n}_{i,j}$ & input-to-terminal gain and command-independent terminal term \\
$\mathbf{C}$ & terminal output-selection matrix \\
$\mathbf{F}^{k},\ \mathbf{g}_{M},\ \mathbf{g}_{T}$ & discrete transition and input matrices \\
$f,\ t_{go}^{k}$ & terminal step and time-to-go, both random \\
$g(\cdot),\ J$ & terminal cost and performance index \\
$\mathbf{K}^{k}$ & interceptor guidance-gain matrix \\
$P^{k}_{j}$ & mode probability of guidance law $p_{j}$ \\
$p_{f}(\cdot)$ & probability mass function of the terminal step \\
$\mathbf{q}_{M},\ \mathbf{q}_{T}$ & interceptor and target internal-dynamics state vectors \\
$S_{n}$ & terminal-set shaping function \\
$u_{M},\ \hat{u}_{M}$ & interceptor acceleration command and its target-side estimate \\
$u_{T},\ u_{T}^{\max}$ & target acceleration command and its bound \\
$\mathbf{x}$ & engagement state vector \\
$z,\ \mathcal{N}$ & zero-effort miss and navigation gain \\
$\Delta t$ & sampling interval \\
$\eta$ & uncertain interceptor-dynamics parameters \\
$\boldsymbol{\mu}^{n}_{i,j},\ \boldsymbol{\Sigma}^{n}_{i,j}$ & mean and covariance of $\mathbf{z}^{n}_{i,j}$ \\
$\xi,\ \dot{\xi}$ & relative displacement and velocity normal to the initial LOS \\
$\tau_{M},\ \tau_{T}$ & interceptor and target dynamics time constants \\
$\omega^{k},\ \nu^{k}$ & process and measurement noise \\
\end{tabular}
\par}

%% file: intro.tex
\section{Introduction}

At the terminal phase of an interception scenario, a defending aerial vehicle must perform evasive maneuvers to avoid an attacking missile. This work considers such engagements in the commonly used planar lateral engagement, where the interceptor uses a state-feedback guidance law to home in on a maneuvering target. The optimal design of interceptor guidance laws has been extensively studied under deterministic settings, leading to canonical strategies such as proportional navigation (PN), augmented PN (APN), the optimal guidance law (OGL)~\cite{zarchan_tactical_2012}, and the linear-quadratic differential game (LQDG) guidance law~\cite{ben-asher_advances_1998}. These strategies arise from linearized engagement dynamics and quadratic cost functionals and thus induce linear acceleration commands.

From the evader's perspective, the optimal maneuver against such linear guidance laws, under bounded acceleration constraints, has been shown to be of bang--bang form, i.e., switching between the minimal and maximal allowable control inputs. This structure was first demonstrated for evasion from PN-guided missiles in both two-dimensional~\cite{shinar_analysis_1977,ben-asher_optimal_1989} and three-dimensional~\cite{shinar_analysis_1979} settings. Later, Shima~\cite{shima_optimal_2011} generalized the result to cover a broad class of linear guidance laws, including PN, APN, and OGL, showing that the evader's optimal response remains bang--bang. These findings were extended to more realistic scenarios by incorporating estimator-based inference. In particular, Fonod and Shima~\cite{fonod_multiple_2016} employed a multiple model adaptive estimator (MMAE)~\cite{shima_efficient_2002} to identify the interceptor's guidance law from a known library and applied the corresponding bang--bang optimal evasion policy. Subsequently, Turetsky and Shima~\cite{turetsky_target_2016} showed that the bang--bang form persists even when the interceptor switches its guidance law midengagement, regardless of whether the switching time is known to the target. Later, Shaferman~\cite{shaferman_near-optimal_2021} demonstrated that bang--bang maneuvers remain optimal even when the target exploits estimation delays in the interceptor's guidance loop~\cite{hexner_temporal_2008}.
Recent works further highlight the prevalence of structured optimal evasion in adversarial and uncertain settings. In particular, Ref.~\cite{du_three-dimensional_2025} analyzed three-dimensional bang--bang maneuvers without knowledge of the interceptor's guidance law based on the kinematics and engagement geometry, while Ref.~\cite{hou_optimal_2025} derived closed-form optimal evasive strategies against general linear-quadratic-optimal guidance laws. However, all of these studies remain deterministic and assume perfect information, in which both agents possess complete knowledge of their own and their opponent's states.

In realistic scenarios, the evader must, instead, estimate the interceptor's state through noisy measurements. This motivates a stochastic formulation of the evasion problem. When modeling such systems, particularly under bounded controls and potentially non-Gaussian uncertainties, it is appropriate to adopt the generalized separation theorem (GST), originally stated by Witsenhausen~\cite{witsenhausen_separation_1971} and derived from Striebel's sufficiency results~\cite{striebel_sufficient_1965}. The GST allows the estimator to be designed separately, whereas the controller uses the posterior distribution of the system state as computed by the estimator.

In stochastic evasion settings, where full information about the states is unavailable, evader behavior is often modeled using randomized or probabilistic control profiles. Commonly used examples include the random telegraph signal (RTS) model, where the target arbitrarily switches between the maximal and minimal acceleration commands~\cite{zarchan_representation_1979,zarchan_tactical_2012}; the Singer process~\cite{singer_estimating_1970}, in which acceleration follows a first-order Gauss–Markov process with finite correlation time; and weaving maneuvers, where the evader performs sinusoidal motion~\cite{zarchan_tactical_2012}. 
Variants of the weaving model have been extensively studied, e.g., in~\cite{yanushevsky_analysis_2004,marks_multiple_2006,weiss_robust_2008,ratnoo_three-point_2016}, often in the context of assessing guidance system robustness or estimator adaptation. Although RTS, Singer, and weaving models are useful for evaluating interceptor performance, they are generally heuristic and do not yield provably optimal feedback evasion strategies under stochastic uncertainty.
Moreover, prior stochastic evasion models lacked a proof of optimal feedback form under uncertainty.

This work addresses that gap by studying the optimal evasion problem
for a bounded-control evader operating under uncertainty about the
engagement's state.  We formulate the problem in discrete time
following the GST lines, where a posterior distribution represents the
imperfect information about the states. Our first and main
contribution is a general existence result: we prove that an optimal
evasion policy exists, and that there is at least one such policy with
a bang--bang structure. This extends the well-known deterministic
result to the stochastic setting and establishes that the structure of
optimal evasive maneuvers remains extremal even under uncertainty.
Importantly, this result reduces the complexity of optimal evasion
synthesis from an infinite-dimensional control space to a finite one,
albeit still requiring a considerable computational cost.

Having proved the main, general result, we then use it to introduce
our second contribution, which is a specific evasion policy based on
the notion of terminal sets.  The idea of using such sets for guidance
problems was first presented in~\cite{shaviv_estimation-guided_2017},
where the pursuer's guidance law was devised using sequential Monte
Carlo (MC) techniques.  Later, Ref.~\cite{mudrik_terminal-set-based_2023}
showed how this approach can be evaluated analytically in a
linearized-Gaussian setting.  This policy leverages the bang--bang
structure to select between the extremals at each time step based on
predicted terminal outcomes derived from state estimates.  We
demonstrate its efficacy in simulation against a proportional
navigation-guided pursuer and benchmark its performance against the
RTS, Singer, and weaving models commonly used in the stochastic
guidance literature.

The remainder of this paper is organized as follows:
Section~\ref{Sec:prob_form} formulates the problem and modeling
assumptions.  The first and main contribution of the paper is
introduced in Sec.~\ref{sec:result}, which formulates and proves
the existence theorem. Section~\ref{sec:terminal_set_evasion} presents
the second contribution of the paper, i.e., the terminal-set-based
evasion (TSE) strategy with conventional stochastic models.
Section~\ref{sec:ts_numerical_eval} presents simulation results of the
TSE, including comparison with the RTS, Singer, and weaving evasion
models.  Concluding remarks are offered in Sec.~\ref{sec:conc}.

%% file: probdef.tex
\section{Problem Formulation}
\label{Sec:prob_form}
A linearized, pursuer--evader stochastic engagement scenario is considered. The problem is mathematically formulated in the sequel.
\subsection{Linearized Kinematics and Dynamics}
Figure~\ref{fig:Planar-engagement-geometry} shows a schematic view of
the geometry of the assumed planar endgame scenario, where $X$-$Y$ is
a Cartesian inertial reference frame.  The intercepting missile and
the target are denoted by $M$ and $T$, respectively.  Variables
associated with the interceptor and the target are denoted by
additional subscripts $M$ and $T$, respectively.  The speed, normal
acceleration, and path angle are denoted by $V$, $a$, and $\gamma$,
respectively.  The slant range between the pursuer and the evader is
$\rho$, and the line-of-sight (LOS) angle, measured from the $X_{I}$
axis, is $\lambda$.  We linearize the trajectories of the interceptor
and the target along the initial LOS.  Commonly used in guidance
strategies, the interceptor and target accelerations normal to the
initial LOS are denoted by $\grave{a}_{M}$ and $\grave{a}_{T}$,
respectively. These accelerations satisfy the following relations:
\begin{align}
	\grave{a}_{M} &= a_{M} \cos (\gamma_{M}^{0} - \lambda^{0}) \\
	\grave{a}_{T} &= a_{T} \cos (\gamma_{T}^{0} + \lambda^{0}).
\end{align}
Variables with superscript $0$ are initial values about which the
linearization is performed.
\begin{figure}[tbh]
	\begin{centering}
		\includegraphics[width=\figwidth]{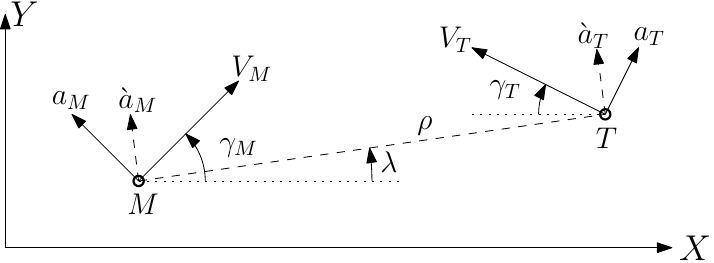}
		\par\end{centering}
	\caption{Planar engagement geometry.}
	\label{fig:Planar-engagement-geometry}
\end{figure}

In this scenario, we assume that the interceptor and the target can be
represented as point masses and that the target's own path angle and
lateral acceleration are known (e.g., via its navigation system).  It
is also assumed that the speeds of both the interceptor and the
target, $V_{M}$ and $V_{T}$, respectively, are known and
time-invariant.  The lateral acceleration command of the target is 
assumed to be bounded by a known constant, $u_{T}^{\max}$.

We define the state vector as
\begin{equation}
	\label{eq:StateVeci}
	\mathbf{x} \rev{\dfn} \begin{bmatrix}
		\xi & \dot{\xi} & \mathbf{q}_{M}^{\top} & \mathbf{q}_{T}^{\top}
	\end{bmatrix}^{\top} \in \mathbb{R}^{n_{x}}
\end{equation}
where \rev{$\xi$ and $\dot{\xi}$ are, respectively, the relative
displacement and relative velocity between the target and the interceptor
normal to the initial LOS, and} $\mathbf{q}_{M}$ and $\mathbf{q}_{T}$ are
the state vectors of internal dynamics of the interceptor and the target,
respectively\rev{, of arbitrary but finite order}.  We assume that their
closed-loop dynamics can be represented via a known, arbitrary-order linear
system.  \rev{In the special case of first-order interceptor and target dynamics,
$\mathbf{q}_{M}$ and $\mathbf{q}_{T}$ reduce to the scalar accelerations
normal to the initial LOS, $\grave{a}_{M}$ and $\grave{a}_{T}$; in
higher-order models they additionally carry the internal autopilot and
actuator states.}  The
discrete-time equations of motion (EOM) are
\begin{equation}
	\mathbf{x}^{k+1} = \mathbf{F}^{k}(\eta) \mathbf{x}^{k} + \mathbf{g}_{M}^{k}(\eta)u_{M}^{k} + \mathbf{g}_{T}^{k}u_{T}^{k} + \omega^{k}
	\label{eq:EOM_evas}
\end{equation}
where $\omega^{k}$ is a white process noise with a known probability density function (PDF) $p_{\omega^{k}}$ and $\eta$ is a random vector representing uncertain parameters of the interceptor dynamics, which also has a known PDF $p_{\eta}$.  We also assume that the process noise and the uncertain parameters are statistically independent.  The process noise is used to stochastically model the effects of physical phenomena that cannot
be modeled deterministically, e.g., turbulence.  The vector of uncertain parameters can be used to handle imprecise knowledge about specific parameters, such as those associated with the interceptor's dynamics, as described in the following example.
\begin{exmp}\label{exmp:first_order}
  Assume that both sides possess strictly proper first-order
  maneuvering dynamics, and define the state vector as
	\begin{equation}
		\label{eq:SV_1st-ord}
		\mathbf{x} \rev{\dfn} \begin{bmatrix}
			\xi & \dot{\xi} & \grave{a}_{M} & \grave{a}_{T}
		\end{bmatrix}^{\top}
	\end{equation}
	\rev{The underlying continuous-time model comprises strictly proper
        first-order lag dynamics for each side},
        {\revblue%
        \begin{subequations}\label{eq:cont_dyn}
          \begin{align}
            \dot{\grave{a}}_{M} &= \frac{u_{M} - \grave{a}_{M}}{\tau_{M}}, \\
            \dot{\grave{a}}_{T} &= \frac{u_{T} - \grave{a}_{T}}{\tau_{T}},
          \end{align}
        \end{subequations}}%
        \rev{together with the relative kinematics normal to the initial LOS,}
        {\revblue%
        \begin{equation}\label{eq:xi_dot}
            \ddot{\xi} = \grave{a}_{T} - \grave{a}_{M}, 
        \end{equation}}%
        \rev{which follow, e.g., \cite[Chap.~2]{zarchan_tactical_2012}.  A
        zero-order-hold discretization} over the interval
        $(t^{k},t^{k}+\Delta t]$ results in the following EOM:
	\begin{equation}
		\rev{\mathbf{F}^{k}} =
		\begin{bmatrix}
                  1 & \Delta t & -\tau_{M}^{2} \psi\big(\frac{\Delta t}{\tau_{M}}\big) & \tau_{T}^{2} \psi\big(\frac{\Delta t}{\tau_{T}}\big) \\
                  0 & 1        & -\tau_{M}\Big(1-e^{-\frac{\Delta t}{\tau_{M}}}\Big) & \tau_{T}\Big(1-e^{-\frac{\Delta t}{\tau_{T}}}\Big)   \\
                  0 & 0        & e^{-\frac{\Delta t}{\tau_{M}}} & 0 \\
                  0 & 0 & 0 & e^{-\frac{\Delta t}{\tau_{T}}}
		\end{bmatrix}, \;\; \mathbf{g}_{M} =
		\begin{bmatrix}
			-\tau_{M}^2 \Upsilon\big(\frac{\Delta t}{\tau_{M}}\big) \\
			-\tau_{M} \psi\big(\frac{\Delta t}{\tau_{M}}\big)       \\
			1-e^{-\frac{\Delta t}{\tau_{M}}} \\
			0
		\end{bmatrix}, \;\;
		\mathbf{g}_{T} = 
		\begin{bmatrix}
			\tau_{T}^2 \Upsilon\big(\frac{\Delta t}{\tau_{T}}\big) \\
			\tau_{T} \psi\big(\frac{\Delta t}{\tau_{T}}\big)       \\
			0                                                       \\
			1-e^{-\frac{\Delta t}{\tau_{T}}}
		\end{bmatrix}
	\end{equation}
	where
	\begin{align}
		\psi(t) & = e^{-t}+t-1, \\
		\Upsilon(t) & = \frac{1}{2} t^{2} - e^{-t} - t + 1.
	\end{align}
	In this case, \rev{the missile's time constant $\tau_{M}$ is, from the
        target's standpoint, an uncertain parameter, whereas the target's
        own time constant $\tau_{T}$ is known to it}.  Assuming that
        $\tau_{M} \in [\tau_{M}^{\min},\tau_{M}^{\max}]$, where
        $\tau_{M}^{\min}$ and $\tau_{M}^{\max}$ are known, we can set
        $\eta \sim U[\tau_{M}^{\min},\tau_{M}^{\max}]$.
\end{exmp}
\subsection{Engagement Duration}
The final time is commonly approximated as
\begin{equation}
	\label{eq:t_f}
	t^{f} \approx \frac{\rho^{0}}{V_{M} \cos (\gamma_{M}^{0} - \lambda^{0}) + V_{T} \cos (\gamma_{T}^{0} + \lambda^{0})}
\end{equation}
where $\rho^{0}$ is the initial range.  This approximation precludes
any deviations from the collision course, and, furthermore, it assumes
that the initial range is perfectly known.  Circumventing the latter
problem, we follow~\cite{rusnak_optimal_2000} and model the final time step \rev{$f \dfn \lfloor t^{f} / \Delta t \rfloor$} as
a random variable with a known probability mass function (PMF),
$p_{f}(i)$ for all $i=0,1,\dots,\infty$.  As the horizon of engagement
scenarios is finite, $f < \infty$ almost surely, yielding that there
exists $i' < \infty$ such that
\begin{equation}
	\sum_{i=i'}^{\infty} p_{f}(i) = 0.
\end{equation}
This allows defining the time-to-go at time $t^{k}$ as
$t^{k}_{go} \dfn \Delta t(f-k)$.  \rev{As $i$ denotes the integer number of steps to go,
its corresponding time-to-go takes the value $t^{k}_{go}=i\,\Delta t$.}  This
variable is also random, and
its PMF satisfies
\begin{equation}\label{eq:ptgo}
  p_{t^{k}_{go}}(i) = \frac{p_{f}(\rev{i + k})}{\sum_{i=0}^{\infty} p_{f}(\rev{i + k})}, \qquad \forall i=0,1,\dots,\infty,
  \qquad \forall k=0,\dots,k',
\end{equation}
where $k'$, the maximal value of $k$ for which Eq.~\eqref{eq:ptgo} is
defined, solves 
\begin{align}
	\max_{k\in \{0,1,\dots,\infty\}} & \; k \notag\\
	\text{s.t. } & \;\sum_{i=0}^{\infty} p_{f}(\rev{i + k}) > 0.
\end{align}
Note that the normalization factor is required in cases where
\begin{equation}
	\sum_{i=0}^{k-1} p_{f}(i ) > 0.
\end{equation}
For any $k$ that nullifies this summation, we can use
$p_{t^{k}_{go}}(i) = p_{f}(\rev{i + k})$ for all
$i=0,1,\dots,\infty$.
\subsection{Estimation}
Following the guidelines of the GST, we assume that a state estimator
is designed independently.  Using noisy measurements, the estimator
outputs a posterior PDF of the state vector, denoted as $p_{{x}^{k}}$,
at every $\Delta t$ seconds, according to
\begin{equation}
	y^{k} =  h(\mathbf{x}^{k}) + \nu^{k}
	\label{eq:meas}
\end{equation}
where $\nu^{k}$ is a white measurement noise with a known PDF $p_{\nu^{k}}$ and $h(\cdot)$ is a given measurement function.  For example, Ref.~\cite{fonod_multiple_2016} uses the MMAE estimator to both estimate the interceptor's state and identify its guidance law when the measurement function is linear and the driving noise is Gaussian.

Because this work focuses on the structure of the control law of the
evader, we assume that the estimation problem has already been solved and, hence, that the posterior PDF of the state
vector~\eqref{eq:StateVeci} is available.
\subsection{Pursuer Guidance Strategy}
\subsubsection{Linear Guidance Laws}
This work considers evasion from a pursuer that uses a linear guidance
law.  Linear laws result from solving linear-quadratic formulations of
optimal control or differential game problems.  These laws are
commonly represented by the structure:
\begin{equation}\label{eq:GL_struc}
	u_{M} = \frac{\mathcal{N}_{p}z_{p}}{t_{go}^{2}},\qquad p \in \mathcal{P} 
\end{equation}
where $\mathcal{N}$ and $z$ are the navigation gain and the
zero-effort miss (ZEM) of each guidance law, respectively, and
$\mathcal{P}$ is the finite set of possible guidance laws used by the
pursuer.  \rev{The ZEM is the miss distance that would result if, from the
current instant onward, the target continued on its present course and the
interceptor applied no further acceleration command. Thus, it is the predicted
relative separation at interception under zero future
control~\cite[Chap.~2]{zarchan_tactical_2012},
\cite[Chap.~4]{ben-asher_advances_1998}.}  Using the state
vector~\eqref{eq:SV_1st-ord}, we present the following noteworthy examples
for the ZEM:
\begin{subequations}\label{eq:classical_laws}
	\begin{align}
		z_{\text{PN}}   & = x_{1} + x_{2} t_{go} \\
		z_{\text{APN}}  & = x_{1} + x_{2} t_{go} + \frac{1}{2} t_{go}^2 x_{4} \\
		z_{\text{OGL}}  & = x_{1} + x_{2} t_{go} - \tau_{M}^2 \psi(t_{go}/\tau_{M}) x_{3} + \frac{1}{2} t_{go}^2 x_{4}\\
		z_{\text{LQDG}} & = x_{1} + x_{2} t_{go} - \tau_{M}^2 \psi(t_{go}/\tau_{M}) x_{3} + \tau_{T}^2 \psi(t_{go}/\tau_{T}) x_{4} 
	\end{align}
\end{subequations}
which are all based on $t_{go}$, and some also require $\tau_{M}$, both of which are uncertain from the standpoint of the target.  
The navigation gain is constant for the PN and the APN guidance laws.
Yet, it is time-varying in the OGL and the LQDG guidance laws.  
The navigation gain used by the interceptor is also uncertain from the standpoint of the target; thus, it can be included in $\eta$.
\rev{The closed-form ZEM expressions in Eq.~\eqref{eq:classical_laws} for the
PN, APN, and OGL laws follow from~\cite[Chaps.~2, 8]{zarchan_tactical_2012},
and the LQDG expression follows
from~\cite[Chap.~4]{ben-asher_advances_1998}.}

Similar to the guidance law proposed in~\cite{shaferman_near-optimal_2021} and based on the structure of the classical guidance laws~\eqref{eq:GL_struc}, we use the following structure for the guidance law of the interceptor:
\begin{equation}
	u_{M}^{k} = \sum_{i=0}^{d} \mathbf{K}^{k-i}(\eta,t_{go}^{k-i}) \mathbf{x}^{k-i}
	\label{eq:uMk_evas}
\end{equation}
where $d$ denotes delayed states, as we also consider the case that the pursuer uses delayed information regarding the state of the target.
This structure is not limited to a specific guidance law, and it can also include switches between different guidance laws, as proposed, e.g., in~\cite{turetsky_target_2016}. 
\subsubsection{Information About the Interceptor's States}
To accurately model the interceptor's guidance law, the target needs
the information that the interceptor has about the entire state
vector~\eqref{eq:StateVeci}, including the relative states, the
interceptor's own states, and the target's states.  Clearly, this
information is inaccessible to the target.  In a similar fashion to
the solutions presented above, the interceptor acquires this
information from noisy and imperfect measurements using a state
estimator.  However, the target only has the statistical information
provided by its own estimator.  Therefore, \rev{because} the true acceleration
command of the interceptor is unknown, the target needs to estimate it
instead.  Moreover, as the target cannot generally know what guidance
law the pursuer uses, we assume that it has to identify it within a
given set of admissible guidance laws.  Therefore, we can model this
estimate as
\begin{equation}
	\hat{u}_{M}^{k} = \sum_{j=1}^{\lvert\mathcal{P}\rvert} P^{k}_{j} \sum_{i=0}^{d} \mathbf{K}^{k-i}(\eta,t_{go}^{k-i}) \hat{\mathbf{x}}^{k-i}
	\label{eq:uMk_evas_est}
\end{equation}
where $\lvert\mathcal{P}\rvert$ is the cardinality of $\mathcal{P}$, $P^{k}_{j} \dfn Pr(p^{k} = j \given \meas^{k})$ is the
mode probability conditioned on the measurement history, $\meas^k$, which can be obtained from a multiple model state estimator, such as the MMAE, and
$\hat{\mathbf{x}}^{k-i} \sim p_{\hat{x}^{k-i}}$ for any $i=0,\dots,d$.
\begin{exmp}
  Consider an interceptor that can select, at each time step, a
  guidance law from the set
  $\mathcal{P}= \{ \text{PN}_{3} ,\text{PN}_{4} , \text{APN}_{3} \}$,
  where $\text{PN}_i$ stands for proportional navigation with
  navigation gain $i$ and $\text{APN}_3$ stands for APN with
  navigation gain 3.  The state estimator of the target provides
  statistical information about the interceptor's states, and the law
  of total probability yields
	\begin{equation}
		p(\mathbf{x}^{k}_{M} \given \meas^{k}) = \sum_{j=1}^{3} \
		P^{k}_{j} p(\mathbf{x}^{k}_{M} \given p^{k} = j , \meas^{k}).
	\end{equation}
	Thus, the guidance law of the interceptor can be modeled as
	\begin{equation}
		\hat{u}_{M}^{k} = \sum_{i=0}^{k'-1-k} \left[ P^{k}_{1} \frac{3 (\hat{x}^{k}_{1 \given 1} + \hat{x}^{k}_{2 \given 1} i)}{i^{2}} + 
		P^{k}_{2} \frac{4 (\hat{x}^{k}_{1 \given 2} + \hat{x}^{k}_{2 \given 2} i)}{i^{2}} + 
		P^{k}_{3} \frac{3 (\hat{x}^{k}_{1 \given 3} + \hat{x}^{k}_{2 \given 3} i + \frac{1}{2} i^2 x^{k-d}_{4} )}{i^{2} }\right]
		p_{t^{k}_{go}}(i)
	\end{equation}
	where
        $\hat{x}_{i \given j} \dfn \expect[x_{i} \given p^{k} = j ,
        \meas^{k}]$ is the current time conditional expectation of the
        i$^{\text{th}}$ state given the measurements and given that
        the j$^{\text{th}}$ mode is valid. Notice that $x^{k}_{4}$,
        which is the target acceleration, is perfectly known to the
        target but not to the interceptor.  Thus, we replace the
        current time estimate with delayed information, $x^{k-d}_{4}$,
        as proposed in~\cite{shaferman_near-optimal_2021}.
\end{exmp}

Determining $p_{\hat{x}^{k-i}}$ is not trivial, as this PDF represents
the target's knowledge about the interceptor's knowledge of the entire
state vector.  The straightforward solution is to use the readily
available outputs of the target's estimator, which are, clearly,
different from the interceptor's estimator's actual outputs.  When
considering the relative states, $\xi$ and $\dot{\xi}$, no better
information regarding these states is available to the target
(however, if it does exist, then the target should use it in its state
estimator).  However, the target's information about the internal
states, $\mathbf{q}_{M}$ and $\mathbf{q}_{T}$, differs drastically
from the interceptor's information.  The target has very accurate
information regarding $\mathbf{q}_{T}$, which is unavailable to the
interceptor.  We can, e.g., add artificial noise or use
delayed data to model the inherent uncertainties of the interceptor's
estimation process.  In contrast, the target's information
regarding $\mathbf{q}_{M}$ is significantly less accurate than the
information the interceptor has, and there is very little the target
can do about that.  Obviously, if better information regarding any of
the states, the dynamics, or the estimator of the interceptor is
available, then it should be used.

These states evolve with time according to
\begin{equation}
	\hat{\mathbf{x}}^{k+1} = \mathbf{F}^{k}(\eta)\hat{\mathbf{x}}^{k} +  \mathbf{g}_{M}^{k}(\eta)\hat{u}_{M}^{k} + \mathbf{g}_{T}^{k}u_{T}^{k} + \hat{\omega}^{k}
	\label{eq:EOM_evas_hat}
\end{equation} 
where $\hat{\omega}^{k}$ is a white process noise with a known PDF $p_{\hat{\omega}^{k}}$, which models the unknown effects the target has about the information of the interceptor.  
For example, the noise about the target's internal dynamics should be higher compared to $\omega^{k}$, as explained above.
\subsection{Performance Index}
\label{subsec:perf_ind}
The cost function, which we aim to maximize in the evasion problem, is set to be
\begin{equation}
	J \dfn \expect \left[ g(\mathbf{x}^{f}) \right]
	\label{eq:J}
\end{equation}
where $g(\cdot)$ is a convex and continuous objective function.  The
expectation is with respect to all the random variables presented
above at the final step of the engagement, $f$, which is also random,
and we assume that it exists.  A detailed analysis of the evasion
problem is presented in the sequel.

%% file: result.tex
\section{Optimal Evasion Structure}
\label{sec:result}
Providing new insight into the existence and structure of optimal
evasion maneuvers in the stochastic scenarios discussed earlier, the
following theorem presents the first and main result of this paper.
\begin{theorem}
  At each time $t^n$, there exists at least one optimal control
  sequence for the evasion problem. Furthermore, of all optimal
  control sequences, at least one must possess the well-known
  bang--bang structure, i.e.,
	\begin{equation}
		\overset{\star}{u}_{T}^{k} \in \{-u_{T}^{\max}, u_{T}^{\max} \}
	\end{equation}
	for all $k=n,\dots,k'-1$.
	\label{th:BB_evas}
\end{theorem}
%
Thus, only the boundaries $-u_{T}^{\max}$ and $u_{T}^{\max}$ of the
interval $[-u_{T}^{\max},u_{T}^{\max}]$ need to be considered at each
time step, instead of the entire interval, which, clearly,
significantly reduces the computational complexity of the optimal
control problem.  
However, even this drastic reduction in computational complexity
cannot bypass the curse of dimensionality, as it leaves us with
exponential complexity.  
Furthermore, we note that Theorem~\ref{th:BB_evas} clearly does not
guarantee the uniqueness of any optimal evasion maneuver (including
the bang--bang one).
\begin{proof}
  \rev{At a high level, the argument rests on a classical result of convex
    optimization: a convex, continuous function attains its maximum over a
    convex, compact set at an extreme point of that set
    (Theorem~\ref{th:max_opt}). We therefore cast the evasion problem in
    this form, establishing that its constraint set is convex and compact
    (Lemma~\ref{lem:set_lemma}) and that its objective is convex and
    continuous over that set (Lemma~\ref{lem:objective}). Since the
    constraint set is a hypercube in the space of the target acceleration
    command sequences, its extreme points are exactly the sequences taking
    the boundary values
    $\pm u_{T}^{\max}$ at every stage, which is the asserted bang--bang
    structure.}

  For convenience of exposition, we first explicitly rewrite the
  evasion problem presented in the previous section.  Ideally, we seek
  to find an evasion law in feedback form, which depends on the
  states, the time-to-go (or the final time of the engagement), and
  the uncertain parameters.  Thus, from the standpoint of feedback
  implementation, this problem is formulated as a stochastic model
  predictive control problem, i.e., we solve it for its entire
  horizon and then implement just the first control action.  Notice
  that this is an infinite-dimensional problem, and its solution is
  generally intractable. The evasion problem at each time step $t^{n}$ is
	\begin{subequations}
		\label{eq:evas_prob}
		\begin{align}
			\underset{\{u_{T}^{k}\}_{k=n}^{k'-1}}{\max} &\; \expect \left[ g(\mathbf{x}^{f}) \right] \\
			\label{subeq:eq_x}
			\text{s.t. } & \; \mathbf{x}^{k+1} = \mathbf{F}^{k}(\eta)\mathbf{x}^{k} +  \mathbf{g}_{M}^{k}(\eta)\hat{u}_{M}^{k} + \mathbf{g}_{T}^{k}u_{T}^{k} + \omega^{k}, \quad \forall k=n,\dots,k'-1, \\
			& \; \mathbf{x}^{n} \sim p_{{x}^{n}}, \quad \eta \sim p_{\eta}, \quad \omega^{k} \sim p_{\omega^{k}}, \quad \forall k=n,\dots,k'-1, \\	
			\label{subeq:eq_u}
			& \; \hat{u}_{M}^{k} = \sum_{j=1}^{\lvert\mathcal{P}\rvert} P^{k}_{j} \sum_{i=0}^{d} \mathbf{K}^{k-i}(\eta,t_{go}^{k-i}) \hat{\mathbf{x}}^{k-i},  \quad \forall k=n,\dots,k'-1,\\
			\label{subeq:eq_xh}
			& \; \hat{\mathbf{x}}^{k+1} = \mathbf{F}^{k}(\eta)\hat{\mathbf{x}}^{k} +  \mathbf{g}_{M}^{k}(\eta)\hat{u}_{M}^{k} + \mathbf{g}_{T}^{k}u_{T}^{k} + \hat{\omega}^{k}, \quad \forall k=n,\dots,k'-1, \\
			& \;  \hat{\mathbf{x}}^{n-i} \sim p_{\hat{x}^{n-i}}, \quad \forall i=0,\dots,d, \quad \hat{\omega}^{k} \sim p_{\hat{\omega}^{k}}, \quad \forall k=n,\dots,k'-1,\\		
			& \; f \sim p_{f}(i), \quad t_{go}^{k} \sim p_{t^{k}_{go}}(i), \quad \forall i=0,1,\dots,\infty, \quad \forall k=n-d,\dots,k'\\
			\label{subeq:u_T_bound}
			&  \; | u_{T}^{k} | \leq u_{T}^{\max}, \quad \forall k=n,\dots,k'-1.
		\end{align}
	\end{subequations}
	
        To prove Theorem~\ref{th:BB_evas}, we rely on established
        findings from convex optimization
        theory~\cite{beck_introduction_2014,boyd_convex_2004}, and, in
        particular, on the following theorem.
	\begin{theorem}[Theorem 7.42 in~\cite{beck_introduction_2014}]
		\label{th:max_opt}
		Let $J: S\rightarrow \mathbb{R}$ be a convex and
                continuous objective function over the convex and
                compact set $S \subseteq \mathbb{R}^{n}$.  Then there
                exists at least one maximizer of $J$ over $S$ that is
                an extreme point of $S$.
	\end{theorem}

	To employ Theorem~\ref{th:max_opt}, we need to prove that the
        objective function of the evasion problem~\eqref{eq:evas_prob}
        is convex and continuous over its constraint set, which is
        convex and compact.  It would then follow that the set of
        optimal evasion maneuvers is nonempty, and, at each time step,
        there exists an optimal solution in the set
        $\{-u_{T}^{\max}, u_{T}^{\max}\}$, i.e., an optimal solution
        having a bang--bang structure.
	
	We proceed by proving the following two lemmas.  The first
        lemma proposes that the constraint set is convex and compact.
        The second lemma establishes that the objective function is
        convex and continuous over the constraint set. The main theorem then follows directly from Theorem~\ref{th:max_opt}
        and these two lemmas.
	\begin{lemma}
          \label{lem:set_lemma}
          Presented in
          Eqs.~\eqref{subeq:eq_x}--\eqref{subeq:u_T_bound}, the
          underlying constraint set of the evasion problem is convex
          and compact.
	\end{lemma}
	\begin{proof}
          The underlying set is presented in its functional form.  If
          the functional form comprises affine functions for equality
          constraints and convex functions for nonpositive inequality
          constraints, then this set, denoted as $S$ in
          Theorem~\ref{th:max_opt}, is
          convex~\cite[Chap.~8.1]{beck_introduction_2014}.  The equality
          constraints~\eqref{subeq:eq_x},~\eqref{subeq:eq_u},
          and~\eqref{subeq:eq_xh} can be substituted recursively to
          eliminate the state variables
          $\mathbf{x}^{n+1},\dots,\mathbf{x}^{k'}$ and the estimated
          guidance actions of the interceptor
          $\hat{u}_{M}^{n},\dots,\hat{u}_{M}^{k'-1}$, and the
          estimated state variables used by the interceptor
          $\hat{\mathbf{x}}^{n+1},\dots,\hat{\mathbf{x}}^{k'}$,
          respectively.  Hence, the terminal state can be presented in
          the following manner:
        \begin{equation}
        \label{eq:xf-affine}
        \mathbf{x}^f(\mathbf{u})\;=\;A^f\,\mathbf{u}+\mathbf{b}^f,
        \end{equation}
        where
        $\mathbf{u} \triangleq [u_{T}^{n},\dots,u_{T}^{k'-1}]^\top$ is
        the vector that contains the target acceleration commands for
        the horizon, and $A^f,\mathbf{b}^f$ depend on the exogenous
        variables, such as $\eta$, $\{\omega^k\}_{n}^{k'-1}$, and
        $\mathbf{x}^n$, but not on $\mathbf{u}$.  Thus, the underlying
        set can be expressed using only the target acceleration
        commands $\mathbf{u}$ and the terminal state $\mathbf{x}^f$.
        After eliminating the state variables, we remain with the
        constraints in Eq.~\eqref{subeq:u_T_bound}, which can be rewritten
        as $\abs{u_{T}^{k}} - u_{T}^{\max} \leq 0$, i.e., the underlying set is the hypercube
        $[-u_T^{\max},u_T^{\max}]^{k'-n}$, which is convex and
        compact.
	\end{proof}
	
	We next prove that the objective function is both convex and
        continuous over the constraint set.
	\begin{lemma}
          \label{lem:objective}
          Defined in Eq.~\eqref{eq:evas_prob}, the objective function
          of the evasion problem is convex and continuous over the
          constraint set.
        \end{lemma}
	\begin{proof}
          Lemma~\ref{lem:set_lemma} shows that the underlying set of
          the evasion problem is convex, and, while proving it, we
          have established that only the target acceleration commands
          $u_{T}^{n},\dots,u_{T}^{k'-1}$ remain as free variables in
          this problem.  Using Eq.~\eqref{eq:xf-affine}, the objective
          function is the expected value of
          $g(\mathbf{x}^f)=g(A^f\,\mathbf{u}+\mathbf{b}^f)$, which
          preserves its convexity under linear change of
          variables~\cite[Theorem~7.17]{beck_introduction_2014}.
          Moreover, the expectation operator preserves the convexity
          of $g(\cdot)$ as infinite summation and multiplication by
          nonnegative scalars also preserve
          convexity~\cite[Sec.~3.2.1]{boyd_convex_2004}.  The lemma
          then follows from observing that these operations also
          preserve continuity.
	\end{proof}

	As we have shown that the objective function is convex and
        continuous (Lemma~\ref{lem:objective}) over its convex and
        compact constraint set (Lemma~\ref{lem:set_lemma}), our main
        result follows directly from Theorem~\ref{th:max_opt}.
\end{proof}

%% file: terminal_set_LM.tex
\section{Terminal-Set-Based Evasion}
\label{sec:terminal_set_evasion}

This section leverages the bang--bang structure established earlier to
obtain a practical evasion law with analytic evaluation based on
terminal sets.  The construction of these sets allows us to evaluate
the resulting cost for the two extreme current commands
\(u_{T}^{n}\in\{\pm u_{T}^{\max}\}\) in closed form.  To make use of
these evaluations, we define a selector as a rule that maps the
predicted terminal outcomes associated with the two admissible
commands, \(u_{T}^{n}=+u_{T}^{\max}\) and \(u_{T}^{n}=-u_{T}^{\max}\),
to the chosen control input at stage \(n\).  The resulting selector
reduces the exponential search over future bang sequences to a single
score comparison that determines the current acceleration command.
\subsection{Scope and Assumptions}
\label{subsec:ts_scope_assumptions}
Consider the discrete-time stochastic endgame formulated in Sec.~\ref{Sec:prob_form}.
We denote \(n\) as the current decision step, \(k\) a generic stage, \(i\) a candidate terminal index, and \(f\) the random terminal index.
The state is \(\mathbf{x}\in\mathbb{R}^{n_x}\).
We let \(\mathbf{C}\in\mathbb{R}^{p\times n_x}\) be a fixed output-selection matrix that extracts the components of the terminal state relevant to the performance measure (e.g., miss distance).
The terminal-time PMF is \(p_f(i)=\mathbb{P}[f=i]\).
The pursuer’s guidance law, \(p_j\), \(j\in\{1,\dots,\lvert\mathcal{P}\rvert\}\), is selected from \(\mathcal{P}\) with mode probabilities \(P^{n}_{j}\) at stage \(n\).
An estimator provides the necessary first and second moments; e.g., the Kalman filter (KF) yields Gaussian posteriors, and the MMAE yields a Gaussian mixture model.
The performance index instantiated here is a squared terminal cost,
\begin{equation}
\label{eq:ts_term_cost_def}
g(\mathbf{x}^{f}) \;=\; \big\| \mathbf{C}\,\mathbf{x}^{f} \big\|^2 ,
\end{equation}
rendering \(g(\cdot)\) a convex and continuous cost, as assumed in
Sec.~\ref{subsec:perf_ind}.
\subsection{Terminal-Set Representation}
\label{subsec:ts_terminal_representation}
For each index $i$ in the support of the random terminal index $f$, i.e., each candidate terminal step with $p_f(i)>0$, and for each admissible guidance-law mode $j \in \{1,\dots,\lvert\mathcal{P}\rvert\}$, unrolling the discrete-time closed-loop dynamics yields an affine dependence of the terminal state on the current command:
\begin{equation}
\label{eq:ts_terminal_affine}
\mathbf{x}^{f}_i \;=\; \mathbf{a}^n(i,j)\,u_{T}^{n} \;+\; \mathbf{z}_{i,j}^{n}.
\end{equation}

The vector $\mathbf{a}^n(i,j)\in\mathbb{R}^{n_x}$ is the input-to-terminal-state gain of the current 
evader command under mode $p_j$,
\begin{equation}
\label{eq:ts_a0_def}
\mathbf{a}^n(i,j)
\;\triangleq\;
\mathbf{\Phi}\!\left(i,n+1;p_j\right)\,\mathbf{g}_{T},
\end{equation}
where the mode-conditioned state-transition product
\begin{equation}
\label{eq:ts_Phi_def}
\mathbf{\Phi}\!\left(i,\ell;p_j\right)
\;\triangleq\;
\mathbf{F}^{i-1}\!\left(\eta,p_j\right)
\cdots
\mathbf{F}^{\ell}\!\left(\eta,p_j\right),
\qquad
\mathbf{\Phi}\!\left(\ell,\ell;p_j\right)=\mathbf{I}_{n_x},
\end{equation}
maps states from stage $\ell$ to $i$. 
Here, $\mathbf{F}^{k}(\eta,p_j)$ denotes the \rev{pursuer-}closed-loop dynamics matrix obtained from the open-loop matrix $\mathbf{F}^{k}(\eta)$ in Sec.~\ref{Sec:prob_form} by substituting the pursuer’s mode-$p_j$ guidance law into the state-update equation.
At the same time, the evader command still enters through $\mathbf{g}_{T} u_{T}^{k}$\rev{, so the loop is closed only around the pursuer's guidance law and not around the target's evasion command}.
Thus, any state delay or time-to-go dependence of the \rev{pursuer's} guidance law is embedded in $\mathbf{F}^{k}(\eta,p_j)$ and, in turn, in $\mathbf{\Phi}(\cdot;p_j)$.

The term $\mathbf{z}_{i,j}^{n}$ aggregates all contributions to the terminal state that do not depend 
on the current command $u_{T}^{n}$:
\begin{equation}
\label{eq:ts_z_def}
\mathbf{z}_{i,j}^{n}
\;\triangleq\;
\sum_{k=n+1}^{i-1} \mathbf{a}^{k}(i,j)\,u_{T}^{k}
\;+\;
\mathbf{d}_{i,j}
\;+\;
\mathbf{w}_{i,j},
\end{equation}
with input-to-terminal gains
\begin{equation}
\label{eq:ts_ak_def}
\mathbf{a}^{k}(i,j)
\;=\;
\mathbf{\Phi}\!\left(i,k+1;p_j\right)\,\mathbf{g}_{T},
\qquad
k=n+1,\dots,i-1.
\end{equation}
The vector $\mathbf{d}_{i,j}$ collects all known, deterministic
contributions to the terminal state under mode $p_j$, including the
pursuer’s closed-loop guidance inputs and any modeled exogenous terms.
The term $\mathbf{w}_{i,j}$ aggregates all stochastic contributions to
the terminal state under mode $p_j$, including the accumulated effects
of process noise, measurement noise, and the propagated
posterior-state uncertainty produced by the estimator (KF or MMAE).
In particular, $\mathbf{w}_{i,j}$ represents the total uncertainty
entering through the noise terms in the dynamics and their propagation
through the estimator’s moment updates, and is, therefore, distinct from
the modeled stochasticity of the future evader inputs
$\{u_T^{k}\}_{k>n}$.

For each $(i,j)$, define the mean and covariance of $\mathbf{z}_{i,j}^{n}$,
\begin{equation}
\label{eq:ts_mu_sigma_def}
\boldsymbol{\mu}_{i,j}^{n}
\;\triangleq\;
\mathbb{E}\!\big[\mathbf{z}_{i,j}^{n}\big],
\qquad
\boldsymbol{\Sigma}_{i,j}^{n}
\;\triangleq\;
\operatorname{Var}\!\big(\mathbf{z}_{i,j}^{n}\big),
\end{equation}
which expand, without requiring independence assumptions, as
\begin{subequations}
    \begin{align}
    \label{eq:ts_mu_expand}
    \boldsymbol{\mu}_{i,j}^{n}
    &=
    \sum_{k=n+1}^{i-1} \mathbf{a}^{k}(i,j)\, m^{k}
    \;+\;
    \mathbf{d}_{i,j}
    \;+\;
    \mathbb{E}\!\big[\mathbf{w}_{i,j}\big], \\
    \label{eq:ts_sigma_expand}
    \boldsymbol{\Sigma}_{i,j}^{n}
    &=
    \sum_{p=n+1}^{i-1} \sum_{q=n+1}^{i-1}
    \mathbf{a}^{p}(i,j)\,
    \operatorname{Cov}\!\big(u_{T}^{p},u_{T}^{q}\big)\,
    \mathbf{a}^{q}(i,j)^{\!\top}
    \;+\;
    \operatorname{Var}\!\big(\mathbf{w}_{i,j}\big),
    \end{align}
\end{subequations}
where $m^{k}\triangleq \mathbb{E}[u_{T}^{k}]$ arises from the
stochastic modeling of future evader inputs in the expectation
evaluation. In particular, the controller optimizes only the current
command $u_T^n$, while the future inputs $\{u_T^k\}_{k>n}$ are treated
as random variables with a prescribed distribution used solely for
evaluating the expected cost.

\subsection{Evasion Law Derivation}
\label{subsec:ts_evas_law}
The objective for the current decision is
\begin{equation}
\label{eq:ts_J_def}
J(u_{T}^{n})
\;\triangleq\;
\mathbb{E}\!\left[\, g(\mathbf{x}^{f}) \,\right]
\;=\;
\sum_{i} p_f(i)\;\sum_{j=1}^{\lvert\mathcal{P}\rvert} P^{n}_{j}\;
\mathbb{E}\!\left[\, \big\|\mathbf{C}\,\mathbf{x}^{f}_i\big\|^2 \,\middle|\, f=i, p_j \right].
\end{equation}
Using Eq.~\eqref{eq:ts_terminal_affine} and the quadratic identity for any
finite second-moment \(\mathbf{z}_{i,j}^{n}\), with \(\|\cdot\|\)
denoting the Euclidean norm, 
\begin{equation}
\label{eq:ts_quad_identity}
\mathbb{E}\!\left[\, \big\|\mathbf{C}\big(\mathbf{a}^n(i,j)u_{T}^{n}+\mathbf{z}_{i,j}^{n}\big)\big\|^2 \,\right]
\;=\;
\big\| \mathbf{C}\,\big(\mathbf{a}^n(i,j)u_{T}^{n}+\boldsymbol{\mu}_{i,j}^{n}\big) \big\|^2
\;+\;
\operatorname{tr}\!\big( \mathbf{C}\,\boldsymbol{\Sigma}_{i,j}^{n}\,\mathbf{C}^{\!\top} \big).
\end{equation}
where \(\operatorname{tr}\) denotes the trace operator.
Note that only the mean term depends on \(u_{T}^{n}\). Therefore, evaluating Eq.~\eqref{eq:ts_J_def} at the two extreme commands gives the terminal-set scores:
\begin{align}
    \label{eq:ts_S_plus}
    & S_{+}^{n}
    \;\triangleq\;
    \sum_{i} p_f(i)\, \sum_{j=1}^{\lvert\mathcal{P}\rvert} P^{n}_{j} \;
    \big\| \mathbf{C}\,\big(\mathbf{a}^n(i,j)u_{T}^{\max}+\boldsymbol{\mu}_{i,j}^{n}\big) \big\|^2, \\
    \label{eq:ts_S_minus}
    &S_{-}^{n}
    \;\triangleq\;
    \sum_{i} p_f(i)\, \sum_{j=1}^{\lvert\mathcal{P}\rvert} P^{n}_{j} \;
    \big\| \mathbf{C}\,\big(-\mathbf{a}^n(i,j)u_{T}^{\max}+\boldsymbol{\mu}_{i,j}^{n}\big) \big\|^2 .
\end{align}
Equivalently, define \(\tilde{\mathbf{a}}_{i,j}^{n}\triangleq \mathbf{C}\mathbf{a}^n(i,j)\) and \(\tilde{\boldsymbol{\mu}}_{i,j}^{n}\triangleq \mathbf{C}\boldsymbol{\mu}_{i,j}^{n}\), then
\begin{equation}
\label{eq:ts_switch_test}
S_{+}^{n}-S_{-}^{n}
\;=\;
4\,u_{T}^{\max}\;
\sum_{i} p_f(i)\, \sum_{j=1}^{\lvert\mathcal{P}\rvert} P^{n}_{j}\;
\big\langle \tilde{\mathbf{a}}_{i,j}^{n},\, \tilde{\boldsymbol{\mu}}_{i,j}^{n} \big\rangle.
\end{equation}
where \(\langle a,b\rangle \triangleq a^\top b\).
Define the shaping function:
\begin{equation}
    \label{eq:ts_shaping_func}
    S_n \;=\; \sum_{i} p_f(i)\, \sum_{j=1}^{\lvert\mathcal{P}\rvert} P^{n}_{j}\;
    \big\langle \tilde{\mathbf{a}}_{i,j}^{n},\, \tilde{\boldsymbol{\mu}}_{i,j}^{n} \big\rangle ,
\end{equation}
which aggregates, over all candidate terminal steps and guidance-law
modes, the expected inner product between the input-to-terminal-state
gain and the corresponding mean terminal shift.  The function $S_n$
encodes, in a single scalar, the net effect of choosing
$u_T^n=\pm u_T^{\max}$ on the expected terminal cost, and its sign
determines which extreme command is preferable.  The obvious optimal
choice is
\begin{equation}
\label{eq:ts_u_star}
\overset{\star}{u}_{T}^{n} \;=\; u_{T}^{\max}\,\sign(S_n) ,
\end{equation}
where $\sign(\cdot)$ denotes the signum function.

We term this proposed evasion law the terminal-set-based evasion strategy and use the shorthand notation TSE in the ensuing
discussion.
\subsection{Discussion}
\label{subsec:ts_discussion}

The terminal-set evaluation is analytical: once the posterior means
and covariances $(\boldsymbol{\mu},\boldsymbol{\Sigma})$ are available
from the estimator, the terminal-set scores
in Eqs.~\eqref{eq:ts_S_plus} and~\eqref{eq:ts_S_minus} and the associated cost
$J(u_{T}^{n})$ follow directly from closed-form expressions.  In contrast to Ref.~\cite{shaviv_estimation-guided_2017}, no sequential MC is
required for online use.  The Gaussian assumption is consistent with
the KF or MMAE posterior, and the latter provides the required
mode-conditioned statistics and weights.

From a computational point of view, the method is lightweight.  Let
$N_f = \#\{\, i : p_f(i) > 0 \,\}$ be the number of terminal steps
with positive probability, $\lvert\mathcal{P}\rvert$ the number of guidance-law
modes, and $\bar H \triangleq \mathbb{E}[\, i - n \,]$ the average
remaining horizon.
Constructing the short-horizon transition products and
input-to-terminal gains scales as \(O(\lvert\mathcal{P}\rvert\,N_f\,\bar H)\) for the
mean terms and \(O(\lvert\mathcal{P}\rvert\,N_f\,\bar H^{2})\) for the covariance
matrices.  Online evaluation of the shaping
function~\eqref{eq:ts_shaping_func} then requires only
\(O(\lvert\mathcal{P}\rvert\,N_f\,n_x)\) operations per decision.  The memory
footprint is modest, limited to storing the precomputed terminal gains
and corresponding moment statistics.

For comparison, the naive bang--bang MPC formulation in Eq.~\eqref{eq:evas_prob} explores, on average, \(2^{\bar H}\) control branches per horizon step, 
yielding exponential complexity in \(\bar H\). 
The terminal-set framework collapses this to near-linear scaling in the horizon, 
achieving orders-of-magnitude savings that make real-time stochastic evasion 
practical.

Finally, the control selector in Eq.~\eqref{eq:ts_u_star} naturally retains the bang--bang structure: maximizing a convex terminal cost over the interval \([ -u_{T}^{\max},\,u_{T}^{\max} ]\) attains an extreme point. 
Uncertainties in both terminal time and guidance law merely reweight the two terminal-set scores without altering this fundamental structure.

%% file: simulation.tex
\section{Numerical Evaluation}
\label{sec:ts_numerical_eval}
We illustrate the terminal-set-based selector on a planar lateral endgame with a single guidance mode. 
\subsection{Simulation Environment}
The evader applies bounded lateral acceleration \(u_{T}^{k}\in[-u_{T}^{\max},u_{T}^{\max}]\) with \(u_{T}^{\max}=9g\). 
The pursuer uses PN with navigation constant \(N=3\). 
The sampling time is \(\Delta t=0.01~\mathrm{s}\), and the closing speed is \(V_c=400~\mathrm{m/s}\).
Let the lateral state be \(\mathbf{x}^k=\begin{bmatrix} \xi^k &  \dot{\xi}^k \end{bmatrix}^{\!\top}\) and
\begin{equation}
\label{eq:ts_lateral_dynamics}
\mathbf{x}^{k+1}
=
\underbrace{\begin{bmatrix}1&\Delta t\\[2pt]0&1\end{bmatrix}}_{\mathbf{F}}
\mathbf{x}^k
+
\underbrace{\begin{bmatrix}\tfrac{(\Delta t)^2}{2}\\[2pt] \Delta t\end{bmatrix}}_{\mathbf{g}_T} u_{T}^{k}
+
\underbrace{\begin{bmatrix}\tfrac{(\Delta t)^2}{2}\\[2pt] \Delta t\end{bmatrix}}_{\mathbf{g}_M} u_{M}^{k},
\qquad
\mathbf{C}=\begin{bmatrix}1&0\end{bmatrix}.
\end{equation}
The terminal index is discretely uniformly distributed over the endgame window:
\begin{equation}
\label{eq:ts_pf_uniform}
f\sim\text{U}\{295,\dots,305\}, \quad N_{f} = 11.
\end{equation}
This implies a final time span \(t^{f}\in[2.95,\,3.05]~\mathrm{s}\) and an associated  initial range uncertainty \(\rho^0\approx V_c\,t^{f}\in[1180,\,1220]~\mathrm{m}\).
For a candidate terminal index \(i\), the PN acceleration is
\begin{equation}
\label{eq:ts_pn_law}
u_{M}^{k}(i) \;=\; N\,\frac{\,\xi^k + t_{{go}}^{k}(i)\,\dot{\xi}^k\,}{\big(t_{{go}}^{k}(i)\big)^{2}},
\end{equation}
where $N=3$, which induces the time-varying closed-loop matrices
\begin{equation}
\label{eq:ts_F_timevarying}
\mathbf{F}^k(i) \;=\; \mathbf{F} + \mathbf{g}_M
\begin{bmatrix}
\displaystyle \frac{N}{(t_{{go}}^{k}(i))^{2}} \;\; & \;\; \displaystyle \frac{N}{t_{{go}}^{k}(i)}
\end{bmatrix},
\qquad k=n,\dots,i-1 .
\end{equation}
We select \(\overset{\star}{u}_{T}^{n}\) by Eq.~\eqref{eq:ts_u_star}, and the corresponding analytic performance is
\begin{equation}
\label{eq:ts_expected_cost_example}
\mathbb{E}\!\left[\xi^2(t^{f})\,\middle|\,\overset{\star}{u}_{T}^{n}\right]
=
\sum_{i} p_f(i)\;
\Big(
\big\| \mathbf{C}\,\big(\mathbf{a}^n(i)\,\overset{\star}{u}_{T}^{n}+\boldsymbol{\mu}_i\big) \big\|^2
+
\operatorname{tr}\!\big(\mathbf{C}\,\boldsymbol{\Sigma}_i\,\mathbf{C}^{\!\top}\big)
\Big).
\end{equation}

\subsection{Single Run: Confirming Bang--Bang Structure}
To emphasize extremality, we model future evader commands as an
i.i.d.\ sequence with a uniform distribution,
\(u_{T}^{k}\sim\mathcal{U}[-u_{T}^{\max},u_{T}^{\max}]\) for
\(k>n\). This changes only the variance to
\(\operatorname{Var}(u_{T}^{k})=(u_{T}^{\max})^2/3\) and yields a
quadratic cost
\begin{equation}
\label{eq:ts_J_of_u_example}
J(u_{T}^{n})
=
\sum_{i}p_f(i)\,
\Big(
\big\|\mathbf{C}\big(\mathbf{a}^n(i)\,u_{T}^{n}+\boldsymbol{\mu}_i\big)\big\|^2
+
\operatorname{tr}\!\big(\mathbf{C}\,\boldsymbol{\Sigma}^{\text{(unif)}}_{i}\,\mathbf{C}^{\!\top}\big)
\Big),
\end{equation}
which is maximized on the convex set \([-\!u_{T}^{\max},u_{T}^{\max}]\) at the endpoints \(\pm u_{T}^{\max}\).

Figure~\ref{fig:ts_single_step_boundary} illustrates the shape of the
expected cost $J(u_{T}^{n})$ under the assumption that future evader
inputs are uniformly distributed. The curve is strictly convex, with
maximum value achieved at the control bounds $u_T^{\max}$. This
demonstrates, analytically and numerically, that the optimal command
is bang--bang, even without explicitly assuming a bang--bang structure, thus validating, for this example, the theoretical extremality result
of Theorem~\ref{th:BB_evas}.
\begin{figure}[t]
  \centering
  \includegraphics[width=\figwidth]{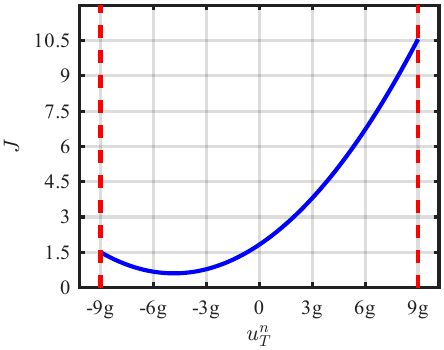}
  \caption{\rev{Expected cost $J(u_{T}^{n})$ under the uniform future-input model.}}
  \label{fig:ts_single_step_boundary}
\end{figure}
\subsection{Monte Carlo Study}
We next conduct an MC study to compare the proposed TSE policy against
three classical stochastic evasion models: the RTS, Singer, and
weaving models.  \rev{These three models are the canonical stochastic
evasion benchmarks in the guidance literature, and they span its three
dominant maneuver paradigms: the RTS captures bang--bang random
switching~\cite{zarchan_representation_1979, zarchan_tactical_2012}, the
Singer process captures correlated Gauss--Markov
acceleration~\cite{singer_estimating_1970}, and the weaving model captures
deterministic periodic motion~\cite{zarchan_tactical_2012}. They are the
standard maneuver models against which interceptor estimators and guidance
laws are evaluated~\cite{yanushevsky_analysis_2004, marks_multiple_2006,
weiss_robust_2008, ratnoo_three-point_2016}, which makes them the natural
points of comparison for the proposed TSE law.}  Each simulation models a
planar lateral engagement
governed by the discrete-time dynamics~\eqref{eq:ts_lateral_dynamics}
with sampling time $\Delta t=0.01~\mathrm{s}$, closing speed
$V_{c}=400~\mathrm{m/s}$, and a proportional navigation pursuer with a
navigation constant $N=3$.  The evader and pursuer accelerations are
limited to $u_{T}^{\max}=9g$ and $u_{M}^{\max}=27g$, respectively.

A total of $N_{\mathrm{MC}}=10{,}000$ independent trials are executed. 
In each trial, the true terminal index $f$ is drawn from the discrete support defined in Eq.~\eqref{eq:ts_pf_uniform}. 
The initial state of the engagement is sampled from a zero-mean Gaussian distribution,
\begin{equation}
    \textbf{x}^{0}\sim\mathcal{N}\!\big(\textbf{0},P^{0}\big), 
    \qquad 
    P^{0}=\operatorname{diag}(100,\,4),
\end{equation}
so that both the initial lateral offset and relative velocity vary randomly across trials. 

The measurement model is
\begin{equation}
    y^{k}= \mathbf{C}\,\mathbf{x}^{k}+\nu^{k},
    \qquad 
    \mathbf{C}=\begin{bmatrix}1 & 0\end{bmatrix},
\end{equation}
where $\nu^{k}\sim\mathcal{N}(0,R_{k})$ represents LOS angle jitter mapped into lateral-position noise. 
The measurement noise variance is not based on the true time-to-go but on a nominal one corresponding to the mean of the terminal-time distribution, ensuring that the estimator operates without access to the actual engagement duration. 
Specifically,
\begin{equation}
    R_{k}=\big(\sigma_{\lambda}V_{c}t_{go}^{k,\mathrm{nom}}\big)^{2},
    \qquad 
    \sigma_{\lambda}=5~\mathrm{mrad},
\end{equation}
where $t_{go}^{k,\mathrm{nom}}=(\bar{f}-k)\Delta t$ and $\bar{f}=\mathbb{E}[f]\approx300$ for the chosen support~\eqref{eq:ts_pf_uniform}. 
The process noise covariance is defined as
\begin{equation}
    Q_{k}=(u_{T}^{\max})^{2}\,\mathbf{g}_{T}\mathbf{g}_{T}^{\!\top}.
\end{equation}

Both sides employ identical KFs as state estimators.  To model its
information advantage, the pursuer is initialized with a smaller covariance:
\begin{equation}
    P_{M}^{0}=\beta P_{T}^{0}, \qquad \beta=0.25,
\end{equation}
indicating that the pursuer's initial uncertainty is smaller than the
evader's.  Moreover, we assume that both agents know each other’s
acceleration commands, which is conservative from the evader’s
perspective. On the target’s side, this assumption is reasonable: as shown in Ref.~\cite{fonod_multiple_2016}, the pursuer’s guidance law can be
identified using an MMAE filter, allowing the target to reconstruct
its acceleration commands. The converse, however, has not been
established, as the authors are not aware of any proof or
demonstration that the pursuer can estimate the target's stochastic
maneuver commands.  Imposing this assumption on the pursuer,
therefore, benefits the pursuer and makes the evasion problem more
challenging for the target, rendering our simulation results
conservative.

During each simulation, the pursuer computes its proportional
navigation acceleration command~\eqref{eq:ts_pn_law} using the
estimates and the mean value for the time-to-go. This command is
saturated to $\pm u_{M}^{\max}$.  The evader acceleration $u_{T}^{n}$
depends on the evasion strategy under test.  For the proposed TSE
case, the terminal-set shaping function $S_{n}$ from Eq.~\eqref{eq:ts_shaping_func} is evaluated at each step, and the
evader applies the bang–bang command~\eqref{eq:ts_u_star}.  
In the RTS case, the evader alternates between $\pm u_{T}^{\max}$, with switching times drawn from an exponential distribution of rate $1/3~\mathrm{s}^{-1}$.
In the Singer case, the evader follows a first-order Gauss–Markov
acceleration process $\dot{a}_{T}=-(1/\tau)a_{T}+w_{a}$ with
$\tau=1~\mathrm{s}$ and $\sigma_{a}=u_{T}^{\max}/2$, truncated to the
control bounds.  In the weaving case, the evader executes a
deterministic sinusoidal maneuver $a_{T}(t)=u_{T}^{\max}\sin(\pi t + \pi/2)$,
producing continuous oscillatory motion.

Figure~\ref{fig:ts_uT_profiles} illustrates the evader acceleration profiles for a representative engagement under all four strategies. \rev{In this figure, the TSE, RTS, Singer, and weaving strategies correspond to the blue, red, magenta, and black lines, respectively.}
For this scenario, the resulting miss distances are $6.29~\mathrm{m}$ for TSE, $5.04~\mathrm{m}$ for RTS, $1.93~\mathrm{m}$ for Singer, and $0.63~\mathrm{m}$ for weaving. 
These differences can be directly related to the qualitative behavior of the evasion laws. 
The bang--bang nature of the TSE law, the randomized sign changes produced by RTS, the correlated stochastic fluctuations generated by the Singer model, and the periodic structure of the weaving maneuver are clearly visible. 
Although TSE, RTS, and weaving all begin and end at the same acceleration values, their switching patterns are fundamentally different. 
In particular, both the presence of switching and the timing of the switch instants play a critical role in shaping the terminal lateral displacement. 
The TSE strategy places its switches at advantageous stages of the engagement to amplify the miss distance, whereas RTS switches at statistically random times and therefore achieves a smaller displacement despite using the same control bounds.

\begin{figure}[t]
  \centering
  \includegraphics[width=\figwidth]{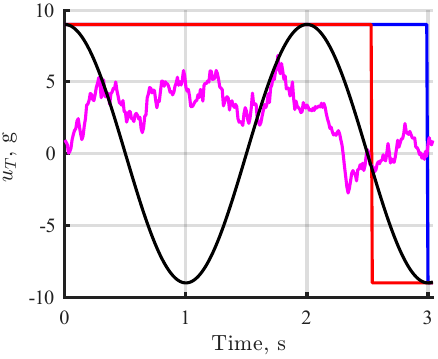}
  \caption{\rev{Evader acceleration profiles for a representative engagement.}}
  \label{fig:ts_uT_profiles}
\end{figure}

Figure~\ref{fig:ts_mc_cdf} compares the empirical CDFs of $|\xi(t^f)|$
across an ensemble of $10{,}000$ trials across the four strategies.  \rev{The curves
use the same color coding as in Fig.~\ref{fig:ts_uT_profiles}.}  A larger
miss distance indicates improved evasion performance; therefore,
higher values correspond to better outcomes for the evader.  The two
bang--bang profiles, TSE and RTS, clearly outperform the stochastic
smooth maneuvers (Singer and weaving), achieving higher miss distances
across nearly the entire distribution. Notably, TSE consistently
dominates RTS, producing larger evasive separations for all quantiles.
For example, suppose the interceptor is equipped with a warhead having
a lethality radius of 1~m. In that case, its single-shot kill
probability (SSKP) against a target using either the Singer or weaving
evasion maneuvers is about 0.9. In contrast, its SSKP against
a target employing the bang--bang RTS evasion model is 0.4, and its
SSKP against a target employing the TSE evasion maneuver is only 0.2
(requiring more than 10 independent interceptors to guarantee a
kill probability of 0.9).
\begin{figure}[t]
  \centering
  \includegraphics[width=\figwidth]{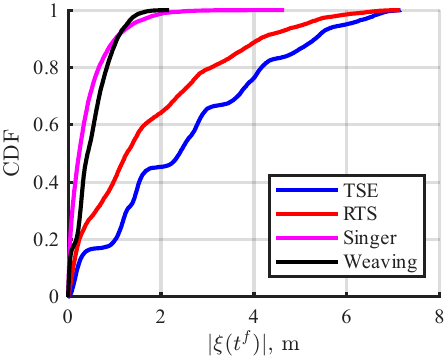}
  \caption{\rev{Empirical CDF of the miss distance over 10,000 Monte Carlo trials.}}
  \label{fig:ts_mc_cdf}
\end{figure}

Table~\ref{tab:ts_mc_stats} reports summary statistics, including the mean, median, and the 5th, 20th, 80th, and 95th percentiles. 
The TSE yields the highest mean and median final miss, along with the widest spread across the 5th to 95th percentiles. This confirms that TSE not only preserves extremality but also surpasses traditional stochastic maneuver models in effectiveness under uncertainty.

\begin{table}[t]
  \centering
  \caption{\rev{Miss-distance $|\xi(t^f)|$ summary statistics in~m for each evasion strategy.}}
  \label{tab:ts_mc_stats}
  \begin{tabular}{lcccccc}
    \hline
    Strategy & Mean & Median & P5 & P20 & P80 & P95 \\
    \hline
    TSE            & 2.55 & 2.4  & 0.13 & 1.02 & 4.16 & 5.98 \\
    RTS            & 1.75 & 1.26 & 0.06 & 0.24 & 3.08 & 5.08 \\
    Singer         & 0.4  & 0.24 & 0.01 & 0.05 & 0.66 & 1.38 \\
    Weaving        & 0.51 & 0.42 & 0.01 & 0.18 & 0.83 & 1.26 \\
    \hline
  \end{tabular}
\end{table}

%% file: conclusions.tex
\section{Conclusions}
\label{sec:conc}
This work addresses the problem of optimal evasion from an interceptor
guided by a linear feedback law in the presence of uncertainty and
bounded controls.  Extending classical results from deterministic to
stochastic settings, we prove that, under the GST guidelines, there
always exists an optimal evasion policy with a bang--bang structure.
This structural result enables a dramatic reduction in the search
space for optimal strategies: rather than optimizing over an
infinite-dimensional input space, it suffices to consider
binary-valued inputs at each time step.

Although this result reduces the formal dimensionality of the problem,
solving for the optimal strategy remains computationally challenging
due to the exponential growth in the number of binary control
sequences with the planning horizon.  To mitigate this, we introduce
terminal-set-based evasion (TSE), a tractable closed-loop method that
exploits the bang--bang structure in a stochastic decision-making
framework.  We evaluate the TSE strategy against three widely used
stochastic evasion models: random telegraph signals, Singer processes,
and weaving.  In comparison, the TSE demonstrates its effectiveness by
generating the largest miss distances.  MC simulation results provide
quantitative evidence of the method’s effectiveness and confirm that
the proposed evasion law achieves the intended improvements in evasion
performance.

\rev{The present formulation assumes a known terminal-time PMF $p_{f}$, and
the TSE law hedges against intercept-time uncertainty only in the
mean-square sense, as implied by the cost~\eqref{eq:J}. A promising direction
for future work is to characterize the sensitivity of the TSE law
and of the RTS, Singer, and weaving benchmarks to misspecification of the
intercept-time distribution and to build in explicit robustness to it, e.g.,
through a chance-constrained or worst-case reformulation that trades
expected miss distance for an improved miss probability~\cite{rusnak_optimal_2000}.}